\documentclass[12pt,oneside]{amsart}
\usepackage{amssymb}
\usepackage{amsmath}
\usepackage{graphicx}
\usepackage{subfigure}
\graphicspath{ {Images/} }
\usepackage{subfig}

\usepackage{float}
\usepackage{multicol}
      \theoremstyle{plain}
      \newtheorem{theorem}{Theorem}[section]
      \newtheorem{lemma}[theorem]{Lemma}

      \theoremstyle{definition}
      \newtheorem{definition}[theorem]{Definition}
      
      \theoremstyle{remark}

      \theoremstyle{example}
      \newtheorem{example}[theorem]{Example}
         \theoremstyle{problem}

         \theoremstyle{proposition}
      \newtheorem{proposition}[theorem]{Proposition}
      
     \theoremstyle{hypothesis}

      \makeatletter
      \def\@setcopyright{}
      \def\serieslogo@{}
      \makeatother

\begin{document}

\title[One-set updates of closure systems]{Direct and Binary Direct Bases for One-set Updates of a Closure System}

\author {Kira Adaricheva}
\address{Department of Mathematics, Hofstra University, Hempstead, NY 11549, USA}
\email{kira.adaricheva@hofstra.edu}

\author {Taylor Ninesling}
\address{Hofstra University, Hempstead, NY 11549}
\email{tninesling1@pride.hofstra.edu}

\thanks{} 
\subjclass[2010]{}

\date{\today}


\maketitle

\begin{abstract}
We introduce a concept of a \emph{binary-direct} implicational basis and show that the shortest binary-direct basis exists and it is known as the $D$-basis introduced in Adaricheva, Nation, Rand \cite{ANR11}. Using this concept we approach the algorithmic solution to the Singleton Horn Extension problem, as well as the one set removal problem, when the closure system is given by the canonical direct or binary-direct basis. In this problem, a new closed set is added to or removed from the closure system forcing the re-write of a given basis. Our goal is to obtain the same type of implicational basis for the new closure system as was given for original closure system and to make the basis update an optimal process.
\end{abstract}

\section{Introduction}
The dynamic update of evolving knowledge bases and ontologies is a routine procedure in the realm of Artificial Intelligence. These applications require tractable representations, such as Horn logic or various versions of descriptive logic. The interest in Horn logic is easily explained by the fact that the reasoning in Horn logic is effective, while the reasoning in general propositional logic is intractable.

If some knowledge base is represented by a (definite) Horn formula $\Sigma$ in variables $X=\{x_1,\dots, x_n\}$, then the set of its models $\mathcal{F}_\Sigma$ forms a lower subsemilattice in $2^X$, which is often referred to as a \emph{closure system} on $X$, or a \emph{Moore family} on $X$. Alternately, one can associate with $\Sigma$ a \emph{closure operator} $\varphi$ on $X$, so that models from $\mathcal{F}_\Sigma$ are exactly the closed sets of $\varphi$. Also, $\Sigma$ can be interpreted as a set of \emph{implications} defining the closure operator $\varphi$. The general connections between Horn formulas (in propositional and first order logic), closure operators and their models were surveyed recently in \cite{AN}.

The knowledge base requires an update if some of the models expire or the new models need to be incorporated into the existing base. In the current work we tackle the problem of re-writing the Horn formula $\Sigma$, when a new model $A$ has to be added or existing model $A$ has to be removed from the family $\mathcal{F}_\Sigma$. 

To avoid misconception, we note that adding set $A$ may result in adding more than just single set to $\mathcal{F}_\Sigma$. Some proper subsets of $A$ may be added as well, which are intersections of $A$ with members of $\mathcal{F}_\Sigma$. This is due to the requirement that the updated family of models must be described by a Horn formula as well, see the classical result in \cite{McK}. 

In the case of the removal of $A$, some more sets have to be removed, if $A$ is the intersection of other sets in $\mathcal{F}_\Sigma$. In this paper, we discuss the case when set $A$ is meet-irreducible in $\mathcal{F}_\Sigma$, thus, only $A$ may be removed.

If the closure operator is encoded through the formal context, the update of the closure system corresponds to adding or removing a row of the table. 

In case of addition of a new set to $\mathcal{F}_\Sigma$, the algorithmic solution for the update of the basis was given in the framework of relational databases in \cite{MR92}, and improvement of the algorithm was suggested in \cite{W95}. The latter publication is the conference proceedings version of the longer and more detailed publication \cite{W94}. Note that in this algorithm the one-set update was considered as one step of iterative process of generating a canonical direct basis of a closure system. The problem was also addressed in the general framework of closure systems, including the FCA framework in \cite{S17} and the framework of Horn-to-Horn belief revision in \cite{ASST}. In the latter paper, the problem was called the Singleton Horn Extension (SHE) Problem. 

In our work we considered two special cases of the SHE problem: when formula $\Sigma$ is given by the \emph{canonical direct basis} of implications defining closure operator $\varphi$, and when it is given by its refined form, the $D$-\emph{basis}. We will assume that one needs an algorithmic solution that provides at the output an updated formula $\Sigma^*(A)$ that is canonical direct, or, respectively, the $D$-basis of the extended closure system.

These two cases will be addressed in sections \ref{CDupdate} and \ref{Dupdate}. Note that one step of iteration process in \cite{W94} also deals with the case of canonical direct basis. Our approach uses a new data structure associated with the basis that allows us to improve the performance time.

In section \ref{Bdirect} we introduce the concept of the binary-direct basis of a closure system and show that the $D$-basis is the shortest binary-direct basis among all the binary-direct implicational bases for the closure system. This allows to extend approach used in the update of the canonical direct basis for the new family of bases, including the $D$-basis.

In section \ref{remove} we present the algorithm of removal one single set from the closure system, assuming that the set is meet-irreducible in the system. The canonical direct basis update is then reduced to the well-known problem of the hypergraph dualization, for which the algorithmic solutions are numerous, see \cite{KBEG06,MU13,FK96}.  

The last section is devoted to the results of algorithmic implementations and testing on various closure systems.

\section{Preparing implicational bases for updates}\label{UpdatePrep}
Prior to the update of the basis, we must do some preparations with regard to the new closed set being added to the family. We assume that the existing basis to update corresponds to a standard closure system (sometimes referred to as a $T\frac{1}{2}$ system \cite{W17}). We choose to operate on the basis representing the equivalent reduced closure system, which is a convenient intermediate form as mentioned in \cite{AN2}. Note that the basis for the reduced closure system will be a superset of the basis for the standard system, but they are equivalent. An algorithm for retrieving a standard system from a reduced system is describe in \cite{AN2}. Since, we are converting from the standard form to the reduced form, we reverse the process which is described as follows:

Say we have a closure system $\langle X, \varphi \rangle$ with corresponding standard system $\langle S, \varphi_S \rangle$. Note that $S \subset X$, and $\varphi_S$ is the restriction of operator $\varphi$ to $S$. If we have $a \in X$ but $a \not\in S$, there is some $B \subset S$ equivalent to $a$. This means that the $\varphi$-closures of these two sets in $\langle X, \varphi \rangle$ are equal, or in terms of implications, $B \to a$ and $a \to b, \forall b \in B$. If $|B| = 1$, we call this a binary equivalence. In the standard system, we remove as many elements from the base set as possible, given all equivalences. However, only the binary equivalences are used to remove elements when deriving the reduced system. So, for each non-binary equivalence $a \leftrightarrow B$, we take $\Sigma' = (\Sigma \cup (B \to a) \cup \{b \to a : b \in B\})^{tr}$. Once we take the expand the basis in this way for each binary equivalence, we obtain the basis corresponding to the desired reduced system.

When a new set is to be added to the basis, some binary equivalences may be broken, and we will need to update the base set and the basis. Say we have some binary equivalence $x \leftrightarrow y$ and want to add new closed set $A$. If $x,y \in A$ or $x,y \not\in A$, the equivalence remains. However, say $x \not\in A, y \in A$. Then we must add to the basis implications, $x \to y$ and $y \to x$. Note that $y$ no longer implies $x$, so $x \to y$ defines the new part of the partial order while we need $y \to x$ so the update process can add in weaker implications that still hold. Additionally, if $x$ was already in our base set, we must add in a copy of each implication containing $x$, replacing $x$ with $y$. After we add these implications, we must again take the transitive closure of the basis. After this is completed for each broken equivalence, we can perform the update procedure.

Our aim for the update is to produce the basis for the updated closure system in its reduced form. However, the basis of the standard closure system is often our desired output. In that case, we can simply apply the aforementioned algorithm from \cite{AN2} to our reduced form. In terms of the basis itself, this process is exactly the process of maximal set projection described after Lemma 13.5 in \cite{MR92} where we project the closed sets into the base set for the standard system.

\section{Update of the canonical direct basis of implications}\label{CDupdate}

In \cite{ASST}, the SHE problem was addressed in the case when
the formula $\Sigma$ describing the knowledge base is assumed to be a conjunction of \emph{prime implicates} of the Horn belief set. Translating this into the language of closure systems, one would call $\Sigma$ \emph{the canonical direct basis}, a type of implicational basis surveyed in \cite{BM10}.

Recall that a formula/implicational set $\Sigma = \{C\rightarrow d: C\cup\{d\} \subseteq X\}$ is called \emph{direct} for closure operator $\varphi$ on $X$, if, for any $Y\subseteq X$,
\[
\varphi(Y) = Y\cup\{d: (C\rightarrow d) \in \Sigma, C\subseteq Y\}.
\]
In other words, the closure of any subset $Y$ can be computed by checking the bodies (left sides) of implications of $\Sigma$ with respect to set $Y$ and expanding $Y$ by their consequents when possible. Each implication of $\Sigma$ is attended only once during the process. Recall that the computation of the closure of $Y$ is generally performed through multiple iteration of $\Sigma$ and expansion of $Y$, see the theoretical background in \cite{W94}, or  through the algorithm known as Forward Chaining \cite{DG84} or LinClosure \cite{MR92}. The canonical direct basis is the smallest implicational set contained in all direct bases defining the same closure operator $\varphi$ on $X$, see \cite{BM10}.

The algorithmic solution for the SHE problem in \cite{ASST} was given in the form of \emph{body-building} formula $\Sigma(A)$, which was produced given a set of implications/formula $\Sigma$ that forms the canonical direct basis of a closure system, and a new set $A$ that needs to be added to the closure system $\mathcal{F}_\Sigma$. We will denote the extended closure system $\mathcal{F}_\Sigma (A)=\mathcal{F}_\Sigma \cup \{F\cap A: F \in \mathcal{F}_\Sigma\}$.

To describe the body-building formula, consider splitting of $\Sigma$ into two subsets: implications $\Sigma_t(A)$ which are true on $A$, and implications $\Sigma_f(A)$ which fail on $A$. 

If $\sigma = (C\rightarrow d) \in \Sigma_f(A)$, then implication fails on set $A$, i.e., $C\subseteq A$ and $d\not \in A$. 
Denote $\sigma(A) =\{C\cup x\rightarrow d: x \in X\setminus (A\cup d)\}$. Note that $\sigma(A) = \emptyset$, if $X\setminus (A\cup d)=\emptyset$. 

Then the body-building formula can be given as 

\[
\Sigma(A)= \Sigma_t(A) \cup \bigcup_{\sigma \in \Sigma_f(A)}\sigma(A)
\]

In other words, the new formula preserves all implications that are true on $A$ and replaces every implication that fails on $A$ by a subset $\sigma(A)$ of new formulas. Each of the new formulas extend the body of a failing implication by a single element not in $A$ and distinct from a consequent.

The formula came up as a consequence to earlier work \cite{LSST}, where the body-building formula was provided to a special extension of the closure system, namely, to the one corresponding to the \emph{saturation operator} $\varphi^*$ associated with given operator $\varphi$. The necessary background for the saturation operator can be found in \cite{CM03}.

In our current work we analyze further the solution for the one-set extension of a closure system. The first observations are collected in the following theorem. We note that item (3) was mentioned in \cite{W94} without proof, and we include the proof for completeness of the exposition.

\begin{theorem} Let $\mathcal{F}_\Sigma$ be a closure system with basis $\Sigma$ and let $A\subseteq X$ be a set not in $\mathcal{F}_\Sigma$. Consider extended closure system $\mathcal{F}_\Sigma (A)$ and body-building formula $\Sigma(A)$.
\begin{itemize}
\item [(1)] Closure system $\mathcal{F}_\Sigma (A)$ comprises the sets that satisfy $\Sigma(A)$.
\item [(2)] Any set $P \not \subseteq A$ that satisfies $\Sigma(A)$ is in $\mathcal{F}_\Sigma (A)$.
\item[(3)] If $\Sigma$ is direct, then $\mathcal{F}_\Sigma (A)$ is defined by basis $\Sigma(A)$, moreover, $\Sigma (A)$ is direct.
\end{itemize}
\end{theorem}

\begin{proof}
The proofs of (1) and (2) are straightforward.

(3) Let $\varphi$ be a closure operator on $X$ corresponding to $\mathcal{F}_\Sigma$.  
Define \[\varphi^*(Y) = \begin{cases}
\varphi(Y) & Y \not\subseteq A \\
\varphi(Y) \cap A & Y \subseteq A \\
\end{cases}\]

Then $\varphi^*$ defines a coarsest closure system that includes all $\varphi$-closed sets and set $A$. Thus, $\varphi^*$ is a closure operator for closure system $\mathcal{F}_\Sigma (A)$.

Define the expected direct closure operator for $\mathcal{F}_\Sigma (A)$, \[ \pi(Y) = Y \cup \{b : (Z \to b) \in \Sigma(A),\ Z \subseteq Y\} \]

First, note that any element of $\pi(Y)$ is either an element of $Y$, or it is the consequent of some $Z \to p \in \Sigma(A)$ corresponding to some $C \to b \in \Sigma$ where $C \subseteq Z \subseteq Y$. Thus, $\pi(Y) \subseteq \varphi(Y)$. If $Y \subseteq A$, and we have some $p \in \pi(Y)$, then $p \in Y \subseteq A$ or there is some $Z \to p \in \Sigma(A)$. However, $p$ must be an element of $A$, otherwise $Z \not\subseteq A$ by the body-building process. So, $\pi(Y) \subseteq \varphi^*(Y)$.

Let $Y \not\subseteq A$. Then, $\varphi^*(Y) = \varphi(Y) = Y \cup \{b : (Z \to b) \in \Sigma,\ Z \subseteq Y\}$ because $\Sigma$ is direct. If $b \in A$, or $Z \not\subseteq A$, then $Z \to b$ is not removed and $b \in \pi(Y)$. Assume $Z \subseteq A$ and $b \not\in A$. Then that implication is removed, but we add $Z \cup \{s\} \to b$ to $\Sigma(A)$ for $s \not\in A \cup \{b\}$. In this case there exists $s \not\in A \cup \{b\}$ such that $s \in Y \setminus A$. Thus, for this particular $s$, $Z \cup \{s\} \subseteq Y$ and once again $b \in \pi(Y)$. So, $\varphi^*(Y) = \pi(Y)$ for $Y \not\subseteq A$.

Now let $Y \subseteq A$.

\[ \begin{array}{lll}
\varphi^*(Y) & = & \varphi(Y) \cap A \\
& = & (Y \cup \{b : (Z \to b) \in \Sigma,\ Z \subseteq Y\}) \cap A \\
& = & (Y \cap A) \cup (\{b : (Z \to b) \in \Sigma,\ Z \subseteq Y\} \cap A) \\
& = & Y \cup ( \{b : (Z \to b) \in \Sigma,\ Z \subseteq Y\} \cap A) \\
\end{array} \]

Consider the set $P = \varphi^*(Y) \setminus Y$. Since $\varphi^*(Y) = Y \cup P$, it suffices to show that $P \subseteq \pi(Y)$. Let $p \in P$. Then $p \in A$, and there is some $Z \to p \in \Sigma$. However, since $p \in A$, $Z \to p \in \Sigma(A)$ and thus $p \in \pi(Y)$. So, $P \subseteq \pi(Y)$, and thus $\varphi^*(Y) \subseteq \pi(Y)$. Since for all $Y \subseteq X$, $\varphi^*(Y) = \pi(Y)$, $\Sigma(A)$ is direct.

\end{proof}
It turns out that without the assumption about the directness of $\Sigma$, the body-building formula $\Sigma (A)$ may lack implications to define the updated closure system. 

\begin{example}
\end{example}
Take base set $X=\{z_1,z_2,z_3,d,u\}$ with the basis $\Sigma =\{z_1z_2 \rightarrow d, z_3 d\rightarrow u\}$ and closure system $\mathcal{F}$. Apparently, this basis is not direct, since it misses the resolution implication $z_1z_2z_3\rightarrow u$.

Consider new set $A=\{z_1,z_2,z_3,u\}$ and consider its subset $Z=\{z_1,z_2,z_3\}$. Since the closure of $Z$ in original system is $X$, $Z$ is not an intersection of $A$ with any set from $\mathcal{F}$, it should not be added when $A$ is added.

On the other hand, the implicational set $\Sigma(A)=\{z_3d\rightarrow u\}$ holds on $Z$. Therefore, $\Sigma(A)$ allows more closed sets than $\mathcal{F}(A)$.

In general, one needs to add to $\Sigma(A)$ implications $C\rightarrow d$ that follow from $\Sigma$ and such that $C\cup d \subseteq A$. Say, in this example, one needs additional implication $z_1z_2z_3\rightarrow u$.

\vspace{0.2cm}

If $\Sigma$ is the \emph{canonical} direct basis, the formula $\Sigma(A)$ may not be the \emph{canonical} direct basis of the updated closure system $\mathcal{F}(A)$.
\begin{example}\label{23}
\end{example}

Indeed, consider $X=\{a,b,c,d,e\}$ and $\Sigma = \{e\rightarrow d, ad\rightarrow e, bc\rightarrow d, abc\rightarrow e\}$. If the new set $A=\{a,b,c\}$, then the body-building formula would require to replace $abc\rightarrow e$ by $abcd\rightarrow e$, but stronger implication $ad\rightarrow e$ is already in $\Sigma$. Similarly, one would need to replace $bc\rightarrow d$ by $bce\rightarrow d$, and $\Sigma$ has stronger implication $e\rightarrow d$. Therefore, $\Sigma(A)$ is not canonical direct basis.

\vspace{0.2cm}

The last example highlights an approach to algorithmic solution to SHE that allows to update the canonical direct basis without the need to reduce implications in the body-building formula. For this we consider a modification of body-building procedure.

For the basis $\Sigma$ and $d \in X$ we will call $\Sigma_d = \{C\rightarrow d\}\subseteq \Sigma$ a $d$-sector of $\Sigma$. With each $C \to d$ in $\Sigma_d$ we will store a list $E_C=\{e_1,\dots e_n\}\subseteq X$ such that $\{e_i\}=E_i\setminus C$ for some $E_i \in \Sigma_d$. Since every implication in $\Sigma_d$ has $d$ as its consequent, we simply store the pair $(C, E_C)$ for each implication in the sector. For Example \ref{23}, the sectors would be $\Sigma_d = \{(bc,\{e\}), (e, \emptyset)\}$, and $\Sigma_e=\{(abc, \{d\}), (ad, \emptyset)\}$.

Consider the modification to the body-building formula. Given basis $\Sigma$ and new set $A$, let $\sigma=(C\rightarrow d)$ be an implication from $\Sigma$, i.e., $(C,E_C)\in \Sigma_d$, and suppose that $\sigma$ fails on $A$.
Define $\sigma^*(A)=\{C\cup x\rightarrow d: x \in X\setminus (A\cup d\cup E_C)\}$. Then modified body-building formula is
\[
\Sigma^*(A)= \Sigma_t(A) \cup \bigcup_{\sigma \in \Sigma_f(A)}\sigma^*(A)
\]

\begin{theorem}
If $\Sigma$ is a canonical direct basis of closure system $\mathcal{F}$, and $\mathcal{F}$ is being extended by new set $A$, then $\Sigma^*(A)$ is the canonical direct basis of $\mathcal{F} (A)$.
\end{theorem}
\begin{proof}
We want to identify implications in direct basis $\Sigma(A)$ which should be deleted to make it canonical direct.

If $C\rightarrow d$ and $G\rightarrow d$ are two implications in $\Sigma$, then $C\not \subseteq G$ and $G\not \subseteq C$, because $\Sigma$ is canonical direct.
If both these implications are in $\Sigma_f(A)$, then $C,G\subseteq A$, and for any $x,y \in X \setminus (A\cup d)$ one has $C\cup x\not \subseteq G\cup y$ and $G\cup y\not \subseteq C\cup x$.

If both implications are in $\Sigma_t(A)$, then they are also in $\Sigma(A)$ without modification.

Thus, the only possibility that one body part is a subset of the other in $\Sigma(A)$ is when $C\rightarrow d$ is in $\Sigma_f(A)$, and $G\rightarrow d$ is in $\Sigma_t(A)$. Thus, we would have $G\subseteq C\cup x$, for some $x \in X \setminus (A\cup d)$. Given that $G\not\subseteq C$, it implies $\{x\}=G\setminus C$, therefore, $x \in E_C$. Thus, we do not need to add implication $C\cup x\rightarrow d$ whenever $x \in E_C$.
\end{proof}

\section{Binary-direct basis of a closure system and the $D$-basis}\label{Bdirect}

This part of the work is devoted to the algorithmic solution for the case when $\Sigma$ is the $D$-basis for the closure operator $\varphi$ and updated formula $\Sigma^*(A)$ is expected to be the $D$-basis of the expanded closure system. 

The $D$-basis was introduced in \cite{ANR11} as a refined and shorter version of the canonical direct basis: the former is a subset of the latter, while the $D$-basis still possessing the form of the directness property, known as \emph{ordered direct} \cite{ANR11}. The closure of any subset $Y$ can be computed attending the implications of the $D$-basis $\Sigma$ only once, when it is done in the specific order. 

The part of the basis containing implications $x\rightarrow y$, i.e. implications with only two variables from $X$, is called \emph{binary}, and it plays a special role in the computation of the closures.

We will assume that the binary part $\Sigma^b$ of basis $\Sigma$  is \emph{transitive}, i.e., if $a\rightarrow b$ and $b\rightarrow c$ are in $\Sigma^b$, then $a\rightarrow c$ is also in $\Sigma^b$.

We will use notation $Y_\downarrow =\{c \in X: (y\rightarrow c)\in \Sigma^b \text{ for some } y \in Y\}$.

Recall that for subsets $Y,Z \subseteq X$ we write $Z\ll Y$ if $Z\subseteq Y_\downarrow$ and say that $Z$ \emph{refines} $Y$.

The following statement describes the relation between the canonical direct and the $D$-basis of a closure system.

\begin{proposition}{\cite{ANR11}}\label{CD} The $D$-basis of a closure system can be obtained from the canonical direct basis $\Sigma_{cd}$ by removing every implication $C\rightarrow d$ for which there exists $D\rightarrow d$ in $\Sigma_{cd}$ with $D\subseteq C_\downarrow$.
\end{proposition}

Note that $D$ cannot be simply an expansion of $C$, because both implications are in the canonical direct basis. Therefore, $|C\setminus D|\geq 1$.

It is well-known that the direct bases of $(X,\varphi)$ are characterized by the property:
\[
\varphi(Y) = Y \cup \{d: (C\rightarrow d)\in \Sigma, C\subseteq Y\}.
\]

In order to characterize the $D$-basis we introduce the following definition.

\begin{definition}
Basis $\Sigma$ for a closure system $(X,\varphi)$ is called b-direct (a shortcut for `binary-direct'), if for every $Y\subseteq X$, $\varphi(Y) = Y_\downarrow \cup \{d: (C\rightarrow d)\in \Sigma, C\subseteq Y_\downarrow\}$
\end{definition}

\begin{proposition}
Any direct basis is b-direct, and every b-direct basis is ordered direct but the inverse statements are not true. 
\end{proposition}

We only mention why the inverse statement are not true. 
As for the first statement, the computation of the closure for a b-direct basis is performed in two stages: first, the binary implications are applied, and then the computation can be done on the expanded set as for the direct basis. In particular, the $D$-basis which is b-direct is not direct. 

For the second statement, we notice that closure systems without cycles may have ordered basis shorter than the $D$-basis, namely, the $E$-basis \cite{ANR11}. According to Theorem \ref{BD} below, the $D$-basis is the shortest b-direct basis, therefore, the $E$-basis is ordered direct but not b-direct.

Thus, the property of being b-direct is stronger than ordered directness. After the binary implications are applied, the order of the remaining implications does not matter, like in a direct basis. This is in contrast with the $E$-basis, for example, where the order of non-binary implications is rather specific.

The following statement is proved similar to Theorem 14 in \cite{BM10}.

\begin{theorem}
Basis $\Sigma$ is b-direct iff $\Sigma^b$ is transitive and for any $A\rightarrow b$ and $C\cup b \rightarrow d$ in $\Sigma$ there exists $G\subseteq (A\cup C)_\downarrow$ such that $G\rightarrow d$ is also in $\Sigma$.
\end{theorem}

The following statement generalizes the description of the canonical direct basis in \cite{BM10}.

\begin{theorem}\label{BD}
Let $(X,\varphi)$ be any closure system.
\begin{itemize}
\item[(1)] 
There exists a smallest b-direct basis, i.e. $\Sigma_{bd}$ such that it is contained in any b-direct basis for a given closure system. 
\item[(2)] $\Sigma_{bd}$ satisfies property : for any two distinct $Z\rightarrow d$, $Y\rightarrow d$ in $\Sigma_{bd}$, $Z\not\subseteq Y_\downarrow$. 
\item[(3)] Basis $\Sigma_{bd}$ is the $D$-basis.
\end{itemize}
\end{theorem}

The proof is done by observing that the $\Sigma_{bd}$ basis can be obtained from the canonical direct as described in Proposition \ref{CD}, therefore it satisfies property in (2), and removing any implication from it will bring to a failure of the property of the b-direct basis.

\section{Algorithm of the $D$-basis update in SHE problem}\label{Dupdate}

Now we describe an effective algorithm for the solution of SHE problem, when the implicational basis $\Sigma$ is the $D$-basis of the associated closure operator $\varphi$.

First, we follow the steps of \ref{UpdatePrep} to prepare the basis for the update and make sure it is in reduced form. Our reasons for working with the reduced system are twofold. First, the $D$-basis describes a partial order $\langle X, \ll \rangle$. A definition for the relation $\ll$ is that for $Y,Z \subseteq X$, $Y \ll Z \text{ iff } \forall y \in Y, \exists z \in Z \ y \in \phi(z)$. If the system is not reduced, then $\ll$ does not satisfy the property of anti-symmetry and is thus not a partial order. Second, the preparation of the basis is easier when only binary equivalences are present. An optimization for the preparation procedure is evident when we have broken equivalence $x \leftrightarrow y$ and we add in all $C \cup y \to d$ such that $C \cup x \to d \in \Sigma$. When $x \to y$ holds, $C \cup y \to d \ll C \cup x \to d$, so we replace $C \cup x \to d$ by $C \cup y \to d$ instead of simply adding them in as in the general case. Now, we perform the actual update procedure.

Let $A$ be a new set that needs to be added to $\mathcal{F}_\Sigma$.

We will denote $\geq$ the partial order imposed on $X$ by the binary part of the basis $\Sigma^b$: $y\geq x$ iff $y\rightarrow x$ is in $\Sigma^b$. We denote $\Sigma_f^b=\{(a\rightarrow c)\in \Sigma^b: a \in A, c \not \in A\}$ the set of binary implications failing on $A$ and $\Sigma^b_t(A)=\Sigma^b\setminus \Sigma_f^b(A)$ the set of all binary implications that hold on $A$.

We define the set of \emph{target} elements $T(A) = \{c \in X\setminus A : (a\rightarrow c)\in \Sigma^b_f \}$. So, target elements are heads of binary implications that fail on set $A$. Whenever target element $c$ is in the body of some implication $C\rightarrow d$ in $\Sigma^{nb}$, another implication with element from $A$ replacing $c$ may become a minimal cover for $d$ in the updated basis $\Sigma(A)$. Thus, we need a process to add such implications to the basis in case they are part of $\Sigma(A)$.

We will denote by $\geq_A$ a partial order imposed on $X$ by $\Sigma^b_t(A)=\Sigma^b(A)$, i.e., the binary part of updated basis $\Sigma(A)$.
Apparently, $y\geq_A x$ implies $y\geq x$. Moreover, the inverse holds when $x\in A$ or $x,y \in X\setminus A$. Similar to $Y_\downarrow$ we use $Y_{\downarrow_A}=\{c: y\geq_A c \text{ for some } y \in Y\}$. We say that implication $Z\rightarrow d$ $\ll_A$-refines $Y\rightarrow d$, or $Z\ll_A Y$, if $Z\subseteq Y_{\downarrow_A}$.\\

The update of $\Sigma$ proceeds in several stages.\\

(I) For each $x \in T(A)$ define $A_x = \{a \in A: a \geq x \text{ and } a \text{ is minimal in } A \text{ with this}\\ \text{property} \}$. Note that $a\not \geq_A x$, for all $a \in A_x$ and that $A_x \subseteq A$. Elements from $A_x$ are replacements of target element $x$, if it appears in the body of any implication in $\Sigma^{nb}$. Also note that, for any $a \in A$ and $x\in X\setminus A$ such that $a\geq x$, there exists $a'\in A_x$ with $a\geq_A a'\geq x$.\\

(II) We will call this part of the procedure $A$-\emph{Lift}, indicating that some new implications $\Sigma^L$ will be added to $\Sigma$ that replace elements $x$ from the bodies of existing implications by elements in $A_x$. The $A$-lift adds implications which may be needed in the body-building phase but have refinements in $\Sigma$.

More precisely, if $C\rightarrow d$ is a non-binary implication, and $C$ has elements from $T(A)$, then we want to add implications $C'\rightarrow d$, when at least one element $x \in C\cap T(A)$ is replaced by some element $a$ from $A_x$. 

We could use notation $\begin{pmatrix} a\\x\end{pmatrix}C''\rightarrow d$ for one instance of $A$-Lift, which records new implication together with element $x\in C$  which is lifted to $a \in A_x$, so that $C=x\cup C''$ and $C'=a\cup C''$.

Note that $|C'|\leq |C|$ and that $C\ll C'$ in the old $\Sigma^b$, but it is no longer true in $\Sigma^b_t(A)$. 
Also note that in the case where $C'' \cap a_{\downarrow A} \neq \emptyset$, we can add a stronger implication than $C'' \cup a$. For example, if we have $C'' = D \cup b$ for $b \in a_{\downarrow A}$, then $C'' \cup a = D \cup b \cup a$, so $D \cup a \ll_A C'' \cup a$. For the $A$-lift, we then want to lift $C' \to d$ to new implications $\begin{pmatrix} a\\x \end{pmatrix}(C'' \setminus a_{\downarrow A}) \to d$.
\begin{example}\label{AliftEx}
\end{example}
Consider $X=\{x,y,d,a,a'\}$ and $\Sigma = \{a\rightarrow x, a'\rightarrow y, xy\rightarrow d\}$. Note that $\Sigma_{cd}$ for the closure system defined by $\Sigma$ will also have $ay\rightarrow d$, $xa'\rightarrow d$ and $aa'\rightarrow d$, but implication $xy\rightarrow d$ refines all three.

If $A=\{a,a'\}$, then  binary implications $a\rightarrow x, a'\rightarrow y$ do not hold on $A$, and the set of targets is $T(A)=\{x,y\}$.

We have $A_x=\{a\}, A_y=\{a'\}$, and implications $\Sigma^L=\{\begin{pmatrix} a\\x\end{pmatrix}y\rightarrow d, x\begin{pmatrix} a'\\y\end{pmatrix}\rightarrow d, \begin{pmatrix} a\\x\end{pmatrix}\begin{pmatrix} a'\\y\end{pmatrix}\rightarrow d\}$ 
are obtained by $A$-Lift from  $xy\rightarrow d$.

Note that implication $aa'\rightarrow d$ in $\Sigma^L$, does not hold on $A$, thus, one needs to modify it further on the body-building stage of the algorithm.\\


(III) Any $(Y\rightarrow d)\in \Sigma^L$ may have a $\ll_A$-refinement in set of implications $\Sigma^{nb}\cup \Sigma^L$.

The following observation shows how to identify implications in $\Sigma^L$ that may have a $\ll_A$-refinement, and thus, can be removed. The refinement may be an original implication from $\Sigma$, of from the $A$-Lift of an original implication.

\begin{proposition}
Suppose $(X\rightarrow d)\in \Sigma^{nb}\cup \Sigma^L$ and $(\begin{pmatrix} a_y\\y\end{pmatrix} Y\rightarrow d) \in \Sigma^L$. If $X\not \ll Y\cup y$ and $X\ll_A Y\cup a_y$, then $a_{y_{\downarrow_A}} \setminus y_\downarrow$ has an element from $X$.
\end{proposition}

\begin{example}
\end{example}
Let us modify Example \ref{AliftEx} by adding element $a_y \in X$, removing implication $a'\rightarrow y$ and adding implications $a_y\rightarrow y$, $a_y\rightarrow a'$ and $xa'\rightarrow d$. Let $A=\{a,a',a_y\}$.

Then $(x\begin{pmatrix} a_y\\y\end{pmatrix}\rightarrow d)\in \Sigma^L$ and it can be refined by $xa'\rightarrow d$. We have $a_{y_{\downarrow_A}} \setminus y_\downarrow = \{a_y, a'\}$, which is in the body of implication $xa'\rightarrow d$.\\

This allows to select implications $X\rightarrow d$ which might be $\ll_A$-refinements of implication $(\begin{pmatrix} a_y\\y\end{pmatrix} Y\rightarrow d) \in \Sigma^L$: compute set $a_{y_{\downarrow_A}} \setminus y_\downarrow$ and check its intersection with $X$. By the definition, $a_{y_{\downarrow_A}} \setminus y_\downarrow\subseteq A$, only implications $X \to d$ such that $X \cap A \neq \emptyset$ should be checked as possible refinements of implications in $\Sigma^L$.

At this stage we remove implications from $\Sigma^L$, if they can be $\ll_A$-refined within $\Sigma_t^b(A) \cup \Sigma^{nb}\cup \Sigma^L$.
We continue using notation $\Sigma^L$ for $A$-lift implications that remain.\\

(IV) In this stage of the algorithm the body-building technique is applied to implications of $\Sigma\cup \Sigma^L$.

Implications from $\Sigma\cup\Sigma^L$ that hold on $A$ will be included in $\Sigma(A)$ without change. 

Now consider $A'\rightarrow d$ in $\Sigma \cup \Sigma^L$ that fails on $A$.

For $A'\subseteq X$, we will use the notation 
$A'_{\uparrow_A} = \{x \in X: x\geq_A a \text{ for some } a \in A'\}$. 

For any body-building by elements in $X\setminus (A\cup d_{\uparrow_A})$ we need to choose $\geq_A$-minimal elements in poset $\langle X\setminus (A\cup d_{\uparrow_A}), \geq_A\rangle$. Indeed, if $x_1\geq_A x_2$ for $x_1,x_2 \in  X\setminus (A\cup d_{\uparrow_A})$, then $A'\cup x_2 \ll_A A'\cup x_1$.

When it happens that minimal element $x_m$ in $\langle X\setminus (A\cup d_{\uparrow_A}), \geq_A\rangle$ also satisfies $x_m\geq_A a$ for some $a \in A'$, then apparently implication $A'\cup x_m\rightarrow d$ can be $\ll_A$-refined to $(A'\setminus a)\cup x_m \rightarrow d$. Therefore, extension by any element $x_m\in A'_{\uparrow_A}$ should be modified to \emph{replacement} of $a$ by $x_m$.

Note that this cannot happen for $A$-Lift $\begin{pmatrix} a\\x\end{pmatrix}A''\rightarrow d$. In this case $x_m\geq_A x$, thus, cannot be minimal in $\langle X\setminus (A\cup d_{\uparrow_A}), \geq_A\rangle$.\\

(V) Body-building may generate implications that can be removed due to $\ll_A$-refinements.
Also note that some implications added by the body-building process may $\ll_A$-refine implications of $\Sigma^L$.

\begin{proposition}\label{refine}
Suppose $x_m$ is minimal in $\langle X\setminus (A\cup d_{\uparrow_A}), \geq_A\rangle$. If $(X'\rightarrow d)\in \Sigma \cup \Sigma^L$ is a $\ll_A$-refinement for $A'\cup x_m\rightarrow d$, then $x_m \in X'$ or $x_m\geq_A a$ for some $a \in X'\cap A$.
\end{proposition}

Indeed, by assumption, $X'\not\ll_A A'$. Then $X'\ll_A A'\cup x_m$ means that $x_m\geq_A x$ for some $x \in X'$. The only element $x\in X\setminus (A\cup d_{\uparrow_A})$ with this property is $x_m$. Since $X'$ cannot have elements in $d_{\uparrow_A}$, we have $x_m\in X$ or $x_m\geq_A a$ for some $a \in X'\cap A$.

Thus, body-building process identifies implications that could be potentially refined.\\





Note that no body-building replacement/extension is a refinement of the other.
Indeed, consider $A'\cup x'\rightarrow d$ and $A''\cup x''\rightarrow d$ extended or replace with minimal elements $x',x'' \in X\setminus (A\cup d_{\uparrow_A})$, and suppose $(A'\cup a')\ll_A (A''\cup a'')$. If $x'=x''$, then $A'\ll_A A''$, but we assumed that all refinements were applied at stage (III). Since $x''\not \geq_A x'$, then $a''\geq_A x'$ for some $a''\in A''$, a contradiction.\\

After applying stages (I)-(V) one obtains implicational set $\Sigma^*(A)$.

\begin{example}
\end{example}
Return again to Example \ref{AliftEx}. Recall that $A$-lift of implication $xy\rightarrow d$ comprises three implications: $ay\rightarrow d$, $xa'\rightarrow d$ and $aa'\rightarrow d$. While the first two belong to $\Sigma^*(A)$, the last one fails on $A$, therefore, it needs body-building update on stage (IV).

These are  $aa'y\rightarrow d$ and $aa'x\rightarrow d$, but both can be refined on stage (V) by 
$ay\rightarrow d$, $xa'\rightarrow d \in \Sigma^L$, respectively. Here, according to Proposition \ref{refine}, minimal elements $x,y$ used for extension are in the bodies of other implications in $\Sigma^L$.

Additionally, the body-building will be used for implications $a\rightarrow x$, $a'\rightarrow y$, giving $ay\rightarrow x$, $ad\rightarrow x$, $a'x\rightarrow y$ and $a'd\rightarrow y$.
\vspace{0.3cm}

\begin{theorem}
$\Sigma^*(A)$ is the $D$-basis of modified closure system $\mathcal{F}(A)$. 
\end{theorem}
\begin{proof}
We assume that $\Sigma$ is the $D$-basis of closure system $\mathcal{F}$ associated with closure operator $\varphi$, thus,
\[
\varphi(Y) = Y_{\downarrow} \cup \{d: (C\rightarrow d)\in \Sigma, C\subseteq Y_{\downarrow}\}.
\]

Let $\varphi_A$ be the closure operator associated with $\mathcal{F}(A)$, i.e., $\varphi_A(Y) = \varphi(Y)$, when $Y\not\subseteq A$, and $\varphi_A(Y)=\varphi(Y)\cap A$ for $Y\subseteq A$.

We want to show that $\Sigma(A)$ defined through the algorithm is b-direct basis for $\varphi_A$, i.e.
\[
(*)\ \ \ \ \ \ \ \ \ \ \ \ \   \varphi_A(Y) = Y_{\downarrow_A} \cup \{d: (C\rightarrow d)\in \Sigma(A), C\subseteq Y_{\downarrow_A}\}.
\]

Indeed, this would show that $\Sigma(A)$ is the b-direct basis for $\varphi_A$; then stages (III) and (V) of the algorithm are intended for the refinement, so their goal is to make the basis $\Sigma(A)$ to satisfy property (2) of Theorem \ref{BD}. As the result, $\Sigma(A)$ would be the shortest b-direct basis, the $D$-basis by Theorem \ref{BD}.

First, we observe that any $\sigma \in \Sigma^L$ holds on $\mathcal{F}$. Similarly, any body-building implication of stage (IV) holds on $\mathcal{F}$. Moreover, $\geq_A\subseteq \geq$, therefore, $Y_{\downarrow_A} \cup \{d: (C\rightarrow d)\in \Sigma(A), C\subseteq Y_{\downarrow_A}\} \subseteq \varphi(Y)$.

If $Y\not \subseteq A$, in order to confirm $(*)$ we would need to show $\varphi(Y) \subseteq Y_{\downarrow_A} \cup \{d: (C\rightarrow d)\in \Sigma(A), C\subseteq Y_{\downarrow_A}\}$. So assume that $y_0 \in Y\setminus A$.

First, show that $Y_\downarrow \subseteq Y_{\downarrow_A} \cup \{d: (C\rightarrow d)\in \Sigma(A), C\subseteq Y_{\downarrow_A}\}$.

Assume that $z \in Y_\downarrow \setminus Y_{\downarrow_A}$. This means $y\geq z$ for some $y \in Y\cap A$ and $z \not \in A$.
On body-building stage (IV) we would add implication $yy_m \rightarrow z$, for some minimal element $y_m$ in $\langle X\setminus A, \geq_A\rangle$ such that $y_0\geq_A y_m$.
Then $yy_m \subseteq Y_{\downarrow_A}$ and $yy_m \rightarrow z$ or its refinement in $\Sigma(A)$. Therefore, $z \in Y_{\downarrow_A} \cup \{d: (C\rightarrow d)\in \Sigma(A), C\subseteq Y_{\downarrow_A}\}$.

Now consider $C\rightarrow d$ in $\Sigma$, where $C\subseteq Y_\downarrow$.
There are two cases to consider. 

1) $C\not \subseteq Y_{\downarrow_A}$, because for some element $c\in C\cap (X\setminus A)$ we have $y\geq c$ for $y\in Y\cap A$, thus, $y\not \geq_A c$.

Take element $c_m\in A$ which is $\geq$-minimal element in $A$ such that $y\geq c_m \geq c$.

Then replacing all such $c \in C$ by $c_m$ will be an $A$-lift implication $C'\rightarrow d$ added on stage (II). Note that $C'\subseteq Y_{\downarrow_A}$.
If $C'\not \subseteq A$, then $C'\rightarrow d$ or its refinement is in $\Sigma(A)$.
If $C'\subseteq A$, then $C'\cup y_m \rightarrow d$ is in $\Sigma(A)$, for some minimal element $y_m$ in $\langle X\setminus A, \geq_A\rangle$ such that $y_0\geq_A y_m$. In either case, $d \in \{d: (C\rightarrow d)\in \Sigma(A), C\subseteq Y_{\downarrow_A}\}$.

2) $C\subseteq Y_{\downarrow_A}$, but $C\subseteq A$, therefore, $C\rightarrow d$ is not in $\Sigma(A)$. In this case $C\cup y_m \rightarrow d$ or its $\geq_A$-refinement is in $\Sigma(A)$, for some minimal element $y_m$ in $\langle X\setminus A, \geq_A\rangle$ such that $y_0\geq_A y_m$. Apparently, $C\cup y_m \subseteq Y_{\downarrow_A}$.

Now consider $Y\subseteq A$. Then $Y_{\downarrow_A} \subseteq A$ and for any $C\rightarrow d$ in $\Sigma(A)$ we must have $d \in A$. Therefore, $Y_{\downarrow_A} \cup \{d: (C\rightarrow d)\in \Sigma(A), C\subseteq Y_{\downarrow_A}\} \subseteq \varphi(Y) \cap A$, and we need only to show the inverse inclusion.

Apparently, $Y_\downarrow \cap A \subseteq Y_{\downarrow_A}$. So we take $d \in A\cap \{d: (C\rightarrow d)\in \Sigma, C\subseteq Y_{\downarrow}\}$. For any $c \in C\cap (X\setminus A)$ we can find $c_m \geq_A c$, a minimal element in $\langle A, \geq_A\rangle$ such that $y\geq_A c_m \geq c$, for some $y \in A$. Replacing all such elements $c$ by $c_m$ makes an $A$-lift $C'\rightarrow d$ of implication $C\rightarrow d$, which is added on stage (II). Apparently, $C'\subseteq Y_{\downarrow_A}$, therefore, $d \in \{d: (C\rightarrow d)\in \Sigma(A), C\subseteq Y_{\downarrow_A}\}$ as desired. 
\end{proof}

\section{Removal of a closed set}\label{remove}

In this section we consider the case when closure system $\mathcal{F}_\Sigma$ on set $X$ defined by the set of implications $\Sigma$ is modified by the removal of one closed set $A \in \mathcal{F}_\Sigma$. The remaining family $\mathcal{F}_\Sigma\setminus \{A\}$ will be again a closure system if and only if $A$ is a meet-irreducible in $\mathcal{F}$: $A=B\cap C$, for some $B,C \in \mathcal{F}$, implies $B=A$ or $C=A$. 

In general, one would need to remove more than just single set $A$, but this can be achieved through the iteration process, where one removes a meet-irreducible set on each step of iteration process. Apparently, there are various paths leading to removal of set $A$.

In the case when the closure system is defined by a context with set of objects $\mathcal{O}$ and set of attributes $\mathcal{A}$, the reduced context with the same closure system on the set of objects $\mathcal{O}^*\subseteq \mathcal{O}$ will contain only rows corresponding the meet-irreducible elements of the closure system. Thus, removal of any row in the \emph{reduced} context corresponds to the removal of meet-irreducible element of closure system.

These considerations prompt to consider the case of meet-irreducible $A \in \mathcal{F}$, which will be our assumption.

We also assume that $\Sigma$ is the canonical direct basis of closure system $\mathcal{F}_\Sigma$. Our goal is to find the canonical direct basis $\Sigma^*$ for 
$\mathcal{F}^*=\mathcal{F}\setminus \{A\}$.

\begin{lemma}
Let $\phi,\phi^*$ be closure operators corresponding to closure systems $\mathcal{F}, \mathcal{F}\setminus \{A\}$, respectively, where $A$ is some meet-irreducible element of $\mathcal{F}$. Then $\phi^*(Y) = \phi(Y)$, for any $Y\not \subseteq A$. 
\end{lemma}

It follows that the closure changes only for subsets of $A$, thus, we should expect new implications $Z\rightarrow x$ with $Z\subseteq A$. To describe premises and consequents of these new implications, we introduce further notations.

Consider (unique) upper cover $A^*$ of $A$ in $\mathcal{F}$, and let $A^*\setminus A=\{ d_1,\dots, d_k\}$.

Recall that a pair $S=\langle B, \mathcal{B}\rangle$, where $B$ is a set and $\mathcal{B}\subseteq 2^B$ is a family of subsets of $B$ is called a \emph{hypergraph}. Subset $U\subseteq B$ is called \emph{a transversal} of hypergraph $S=\langle B, \mathcal{B}\rangle$, if $U\cap V \not = \emptyset$, for every $V\in \mathcal{B}$.

Define hypergraph $H=\langle A, \mathcal{T}\rangle$ as follows:
$C \subseteq A \in \mathcal{T}$ iff $C=A\setminus \phi(A')$, for some $A'\subseteq A$.
In other words, $\mathcal{T}$ is the collection of complements (in $A$) of $\phi$-closed subsets of $A$.

Let $Y_1,\dots, Y_t$ be minimal (with respect to $\subseteq$ relation) transversals of hypergraph $H=\langle A, \mathcal{T}\rangle$.

\begin{lemma}\label{L1}
If $Y\rightarrow d$ is any implication that holds on $\mathcal{F}^*=\mathcal{F}\setminus \{A\}$ and fails on $A$, then $Y_j\subseteq Y$, for some $j\leq t$, and $d=d_i$, for some $i\leq k$.
\end{lemma}
\begin{proof}
It is clear that for any implication $Y\rightarrow d$ failing on $A$ we have $Y\subseteq A$ and $d \not \in A$. Since it holds on $A^*$, we must have $d\in A^*$, i.e. $d=d_i$ for some $i\leq k$.

Apparently, $Y\rightarrow d$ holds on set $Z \in \mathcal{F}^*$, if $d\in Z$. So we consider $Z\in \mathcal{F}^*$ such that $d \not\in Z$. Then $Y\rightarrow d$ holds on $Z$, when $Y\not \subseteq Z$, which also implies $Y\not \subseteq Z\cap A$. Thus,  we only need to consider the case of closed sets $Z\subseteq A$. Then $Y\not \subseteq Z$ iff $Y\cap A\setminus Z \not = \emptyset$. Since it is true for any $Z\in \mathcal{F}^*$, $Z\subseteq A$, we conclude that $Y$ is a transversal of hypergraph $G=\langle A, \mathcal{T}\rangle$. It follows that $Y$ contains some minimal transversal $Y_j$, $j\leq t$. 
\end{proof}

\begin{lemma}
Let $\Sigma_1=\{Y_j\rightarrow d_i: j\leq t, i\leq k\}$. Then $\Sigma \cup \Sigma_1$ is a direct basis of $\mathcal{F}^*$.
\end{lemma}
\begin{proof}
If $Z\rightarrow t$ is any implication that holds on $\mathcal{F}^*$, then it is either holds on $A$ or fails on $A$. If it holds on $A$, then it holds on $\mathcal{F}$, thus, there exists $Z'\rightarrow x \in \Sigma$ such that $Z'\subseteq Z$, since $\Sigma$ is direct basis for $\mathcal{F}$. Thus, $Z\rightarrow t$ follows from $\Sigma$.

If $Z\rightarrow t$ fails on $A$, then according to Lemma \ref{L1}, $t=d_i$, for some $i\leq k$, and $Z\subseteq Y_j$, for some $j\leq y$. Therefore, $Z\rightarrow t$ follows from $\Sigma_1$.

This shows that $\Sigma \cup \Sigma_1$ is a basis for $\mathcal{F}^*$.

In order to show that this is a direct basis, we may use Theorem 14 from \cite{BM10} to prove that basis $\Sigma \cup \Sigma_1$ has the following property: if $Y\rightarrow d$ and $D\cup d\rightarrow x$ are in $\Sigma \cup \Sigma_1$, then there exists $Z \subseteq Y\cup D$ such that $Z\rightarrow x$ is again in $\Sigma \cup \Sigma_1$.


Assume that one of $Y\rightarrow x$, $D\cup x\rightarrow d$ is in $\Sigma$ and the other is in $\Sigma_1$. Since both implications hold on $\mathcal{F}^*$, the resolution $Y\cup D \rightarrow d$ also holds on $\mathcal{F}^*$. Then depending on whether it holds on $A$ or not, there should be $Z\rightarrow x$ in $\Sigma$ or $\Sigma_1$ with $Z\subseteq D\cup Y$.

It is not possible that both implications are from $\Sigma_1$, and if they are both in $\Sigma$, the required conclusion follows from the directness of $\Sigma$.
\end{proof}

Note that $\Sigma\cup\Sigma_1$ is not necessarily the \emph{canonical} direct basis of $\mathcal{F}^*$.

\begin{example}

\end{example}
Consider closure system on $X=\{m_1,m_2,x,d\}$ defined by set of implications $\Sigma = \{d\rightarrow x, m_1m_2x\rightarrow d\}$. Set $A=\{m_1,m_2\}$ is closed and meet-irreducible, its upper cover is $X$. If we want to remove it, then the hypergraph is $H=\langle A, \{\{m_1\}, \{m_2\}\}\rangle$, thus, it has a unique minimal transversal $\{m_1,m_2\}$. Then we need to add implications $m_1m_2\rightarrow x$ and $m_1m_2\rightarrow d$. It follows that original implication $m_1m_2x\rightarrow d$ can be removed.

\section{Algorithmic solutions and testing}

We will present the results of code implementations of two algorithms discussed in section \ref{CDupdate} and \ref{Dupdate} and their testing on some bench-mark examples developed in earlier code implementations for the $D$-basis outputs. In \cite{R16}, the algorithm produces the $D$-basis, when the input is some set of implications defining the closure operator. In \cite{AN17}, the closure operator is encoded in the context, and the extraction of the $D$-basis is done by reduction to known solutions of the hypergraph dualization problem.

We will present the bounds of algorithmic complexity and compare them with the actual time distributions based on parameters  such as the sizes of the input and output.

\section{Canonical Direct Update Implementation}

When we consider the complexity of the update for the Canonical Direct basis, we compare the modified algorithm to a naive implementation of the original body building formula described in \cite{ASST}. When we apply the body building formula, the first step is to remove and replace the implications which are invalid, given the new closed set we wish to add. Going through the basis once, each implication is either kept or removed. If an implication is removed, we add implications where the premise is extended by elements of the base set. If the basis contains $n$ implications and the base set has size $x$, then the application of this formula has complexity $O(n \cdot x)$.

Once this pass is completed, the resulting basis is a direct basis for the desired closure system. However, there may be extra implications which must be eliminated. Let the number of implications in the updated basis be $m \leq n \cdot x$. To remove extra implications, each implication is compared to each other implication in the basis, and if there is a stronger implication with the same consequent, the weaker implication is removed. Since this compares each implication to each other, this step of the process has complexity $O(m^2)$. Once the basis is reduced to be minimal with respect to this condition, we have produced the Canonical Direct basis for the updated closure system.

With the application of the modified body building formula, we follow the same steps for the first part, but the use of $d$-\emph{sectors} allows us to restrict the addition of extra implications. So, the produced basis is the Canonical Direct basis, produced with $O(n \cdot x)$ complexity and has $m$ implications. The added complexity to this comes in building the enriched data structure for the new basis. We keep track of singleton skewed differences between implication premises among sectors. So in each sector, each implication is compared to each other and if the set difference between the premises is a singleton, we store in in the data structure. This must be done for each affected sector. We must update d-sectors where d is not an element of the new closed set to add. In the worst case, we have all $m$ implications in a single sector and they must all be prepared, so our complexity for this step is also $O(m^2)$.

While the two algorithms have the same worst case complexity, the use of $d$-sectors provides practical time benefits. The worst case is that all implications are in a single sector, it is common for implications to be split among several or many sectors.

To test the practical difference between the two algorithms, the algorithms were implemented in Scala and run on random examples. For each of the 1000 examples, a random binary table between the size of $10 \times 10$ and $15 \times 15$ is generated. Then, the Canonical Direct basis of the table is generated and a random subset of the base set is chosen as the new closed set for the update. Then the update is performed with each of the algorithms and compared. 

We also compared the algorithm to an implementation of Algorithm 3 in M. Wild \cite{W94, W95}, which also mentions earlier implementation in H.~Mannila and K.J.~R\"{a}ih\"{a} \cite{MR92} for computation of a direct basis. While the goal of this algorithm is to generate an entire basis from a given Moore family, the algorithm does so iteratively with a subroutine which performs the same action as our algorithm. The algorithm is similar to the original body building formula, the main difference being that the algorithm in \cite{W94} is implemented for the Canonical Direct basis of full implications as opposed to the Canonical Direct unit basis.

\begin{table}[h!]
\caption{Average Update Times by Number of Broken Implications}
\centering
\begin{tabular}{|c|c|c|c|}
\hline
Broken & Naive Update (ms)  & Wild Update (ms) & Modified Update (ms) \\ [0.5ex]
\hline
10 & 91.1 & 23.9 & 18.9 \\
20 & 190.8 & 45.7 & 30.4 \\
40 & 194.6 & 51.8 & 28.3 \\
\hline
Overall & 198.5 & 54.2 & 27.7 \\
\hline
\end{tabular}
\end{table}

Comparatively Wild's algorithm performs much better than the original body building formula, likely due to the smaller number of implications in his condensed basis. Overall, our modified body building formula outperformed Wild's algorithm, especially in cases where the number of broken was large. The modified body building formula completed the update faster than the naive update in all 1000 examples. However, Wild's update algorithm, while slower on average, outperformed the modified formula in 133 cases. Of these cases, 110 occurred when at most 20 implications were broken. The number of broken implications in these examples had an average of 30 and ranged from 0 to 140. 

\begin{center}
\includegraphics[width=4in]{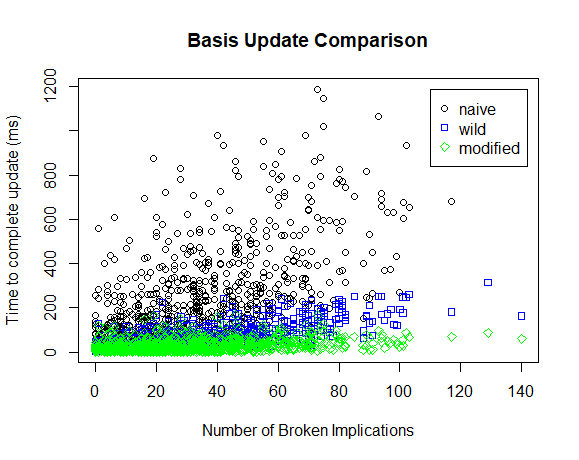}
\end{center}

As we see in the plot of the example data, the data generally follows 3 bands, one for each algorithm with the modified algorithm generally having the lowest update time. 

\vspace{0.2 in}

{\it Acknowledgements.} The first results of the paper were presented on the poster session of ICFCA-2017 in Rennes, France, and both authors' participation in the conference was supported by the research fund of Hofstra University. We thank Sergey Obiedkov for pointing to the important publication of Marcel Wild \cite{W94}, and we thank Justin Cabot-Miller for his support in producing valuable test cases in the implementation phase.

\end{document}